\documentclass[a4paper,draft=true,DIV=classic,fontsize=10pt]{scrartcl}
\KOMAoptions{DIV=15}

\usepackage[latin1]{inputenc}
\usepackage{amsmath,amssymb,amsthm}
\usepackage{enumerate}
\usepackage{arydshln}
\usepackage[usenames,dvipsnames]{color}
\usepackage[left=3.5cm,top=3cm,right=3.5cm,bottom=5cm]{geometry}
\usepackage{setspace}

\usepackage{graphicx}
\usepackage{placeins}
\usepackage{float}
\usepackage{subfigure}

\newtheorem{theorem}{Theorem}[]

\newtheorem{corollary}[theorem]{Corollary}
\newtheorem{assumption}{Assumption}

\newtheorem{remark}{Remark}

\newcommand{\R}{\mathbb{R}}

\begin{document}

\title{The amazing power of dimensional analysis: Quantifying market impact\thanks{All authors acknowledge support by the Vienna Science and Technologie Fund (WWTF) through project MA14-008.
W. Schachermayer is additionally supported by WWTF project MA16-021 and by the Austrian Science Fund (FWF) under
grant P25815 and P28861.}}
\author{Mathias Pohl\thanks{University of Vienna, Faculty of Mathematics and Faculty of Business, Economics \& Statistics, Oskar-Morgenstern-Platz 1, 1090 Vienna, Austria, mathias.pohl@univie.ac.at} \and Alexander Ristig\thanks{University of Vienna, Faculty of Mathematics and Faculty of Business, Economics \& Statistics, Oskar-Morgenstern-Platz 1, 1090 Vienna, Austria, alexander.ristig@univie.ac.at} \and Walter Schachermayer\thanks{University of Vienna, Faculty of Mathematics, Oskar-Morgenstern-Platz 1, 1090 Vienna, Austria, walter.schachermayer@univie.ac.at} \and Ludovic Tangpi\thanks{University of Vienna, Faculty of Mathematics, Oskar-Morgenstern-Platz 1, 1090 Vienna, Austria, ludovic.tangpi@univie.ac.at}}
\date{\today}

\maketitle

 \begin{abstract}
 \noindent 
 {\textbf{Abstract}.} 
This note complements the inspiring work on dimensional analysis and market microstructure by Kyle and Obizhaeva \cite{kyle2017dimensional}. Following closely these authors, our main result shows by a similar argument as usually applied in physics the following remarkable fact. If the market impact of a meta-order only depends on four well-defined and financially meaningful variables, then -- up to a constant -- there is only one possible form of this dependence. In particular, the market impact is proportional to the square-root of the size of the meta-order.

This theorem can be regarded as a special case of a more general result of Kyle and Obizhaeva. These authors consider five variables which might have an influence on the size of the market impact. In this case one finds a richer variety of possible functional relations which we precisely characterize. We also discuss the analogies to classical arguments from physics, such as the period of a pendulum.

\textbf{Keywords}: Dimensional analysis; market impact; leverage neutrality.
\end{abstract}

\bigskip


\section{Introduction}\label{sec:intro}

Dimensional analysis is a well known line of arguments in physics. The idea is best explained by considering a classical example: The period of a pendulum.

The \emph{basic assumption} is that the period depends \emph{only}  on the following quantities:
\begin{itemize}
 \item the length $l$ of the pendulum, measured in meters,
 \item the mass $m$ of the pendulum, measured in grams,
 \item the acceleration $g$ caused by gravity, measured in meters per second squared.
 \end{itemize}
 The basic assumption amounts to the formula,
 \begin{equation}
 \label{eq:period}
 	\text{period} = f(l,m,g),
 \end{equation}
where the period is measured in seconds and $f$ is an -- a priori -- arbitrary function.

Of course, relation \eqref{eq:period} should not depend on whether we measure length by meters or inches, time by seconds or minutes, and mass by grams or pounds.
 Combining these three requirements with the \emph{ansatz}
 \begin{equation}\label{eq:ansatz}
 f(l,m,g) = \text{const} \cdot l^{y_1}m^{y_2}g^{y_3},
 \end{equation}
these requirements translate into three linear equations in the variables $y_1, y_2, y_3$.
The unique solution yields the well-known relation (see \cite{huntley1967dimensional} as well as Appendix~\ref{app:pendulum} below for the details)
\begin{equation}
\label{eq:sqrt-pendulum}
\text{period} = \text{const}\cdot \sqrt{\frac{l}{g}}.
\end{equation}
This result goes back as far as Galileo.
The elementary linear algebra used in the above argument has been formalized in proper generality in the nineteenth century and is known under the name of ``Pi-Theorem'' (see Section~\ref{sec:dim_anal} below). It is worth mentioning that in the present case, the \emph{ansatz} \eqref{eq:ansatz} does not restrict the generality of the solution \eqref{eq:sqrt-pendulum} (see Appendix~\ref{app:pendulum} below).
\newline

Kyle and Obizhaeva have applied this line of argument in \cite{kyle2017dimensional} to analyze the \emph{market impact} of a \emph{meta-order}: think of an investor who wants to buy (or sell) a sizeable amount of an underlying stock within a limited time (e.g. two days). Of course, when placing this \emph{meta-order} she will split it into smaller pieces, the actual orders, in some (hopefully) clever way. Nevertheless, we expect the quoted prices to move to the disadvantage of the agent. We call the expected size of this price movement, measured in percentage of the price, the \emph{market impact}, see \cite{bouchaud2010price}.

We start by identifying the variables (and their dimensions $[\cdot]$) which we expect to have an influence on the size of the market impact:
\begin{itemize}
	\item $Q\quad\,$ the size of the meta-order, measured in units of shares $[Q]=\mathbb{S}$,
	\item $P\quad\,$ the price of the stock, measured in units of money per share $[P]=\mathbb{U}/\mathbb{S}$,
	\item $V\quad\,$ the traded volume of the stock, measured in units of shares per time $[V]=\mathbb{S}/\mathbb{T}$,
	\item $\sigma^2\quad$ the squared volatility of the stock, measured in percentage of the stock price per unit of time $[\sigma^2]=\mathbb{T}^{-1}$.
\end{itemize}

These 4 variables are measured in the units of the 3 fundamental dimensions time $\mathbb{T}$, money $\mathbb{U}$ and shares $\mathbb{S}$. Now we formulate the following basic assumption.

\begin{assumption}\label{ass:4unknown}
The market impact $G$ depends \emph{only} on the above $4$ variables, i.e.
\begin{align}\label{eq:G}
	G &= g(Q,P,V,\sigma^2),
\end{align}
where the function $g:\R_+^4\rightarrow\R_+$ as well as the quantity  $G$ are invariant under changes of the units chosen to measure the ``dimensions'' $\mathbb{T}$, $\mathbb{U}$, $\mathbb{S}$.
\end{assumption}

We note that $G$ is a percentage of the quoted price of the stock; hence it is a ``dimensionless'' quantity, i.e. invariant under a change of the units in which $\mathbb{T}$, $\mathbb{U}$ and $\mathbb{S}$ are measured. We thus encounter an analogous situation as in the pendulum example. There is, however, a serious difference to the pleasant situation encountered above: We now have 4 variables, namely $Q$, $P$, $V$ and $\sigma^2$, but only 3 equations resulting from the scaling invariance for the fundamental dimensions $\mathbb{T}$, $\mathbb{U}$, $\mathbb{S}$. We need one more equation to obtain such a crisp result as in \eqref{eq:sqrt-pendulum} above. Kyle and Obizhaeva found a remedy; an additional ``no arbitrage'' type argument which can be deduced from transferring the Modigliani-Miller invariance principle to market microstructure. To fix ideas, consider a stock which is a share of a company. Suppose that the company changes its capital structure by paying dividends or, passing to the opposite sign, by raising new capital. The Modigliani-Miller theorem precisely tells us which quantities \emph{remain unchanged} when varying the leverage in terms of the relation between debt and equity of the company. This insight should furnish one more equation to be satisfied by \eqref{eq:G}. For the details we refer to Section~\ref{sec:laws} below. The subsequent assumption hints at this additional restriction which Kyle and Obizhaeva call ``leverage neutrality'' and is quoted from Proposition 1 in the seminal paper by Modigliani and Miller \cite{modigliani1958cost} (see Assumption~\ref{ass:MMunit} below for a more formal definition).

\begin{assumption}[Leverage neutrality]\label{ass:mm}
	 The market value of any firm is independent of its capital structure.
\end{assumption}

It turns out that this invariance indeed provides one more linear equation analogous to the equations obtained by the scaling arguments above. We therefore find ourselves in a perfectly analogous situation as with the pendulum and have the same number of equations as unknowns, namely four.

\begin{theorem}\label{thm:SqrtLawIntro}
	Under Assumptions~\ref{ass:4unknown} and \ref{ass:mm}, the market impact is of the form
	 	\begin{align}\label{eq:SqrtLawIntro}
			G &=\emph{const}\cdot \sigma\sqrt{\frac{Q}{V}},
		\end{align}
	 for some constant $\emph{const}>0$.
\end{theorem}

In particular, we find the square-root dependence of the market impact on the order size $Q$ in accordance with several theoretical as well as empirical findings (see the review of the literature below).
\newline

In fact, the above line of arguments does not correspond exactly to what Kyle and Obizhaeva have done in \cite{kyle2017dimensional} (compare also \cite{kyle2016market}). They have considered one more variable which may have influence on the market impact. These authors suppose that the agent faces a cost $C$ when preparing the placement of a meta-order, which the authors refer to as ``bet cost''.\footnote{In the version of \cite{kyle2017dimensional}, released in July 2017, $C$ is defined as the unconditional expected dollar costs of executing a bet, i.e. meta-order.} This ``bet cost'' $C$ may vary independently  of the order size $Q$ as well as of the quantities $P$, $V$ and $\sigma^2$ discussed above. Hence, they consider an additional fifth quantity which might influence the market impact:
\begin{itemize}
	\item $C\quad\,$ the ``bet cost'', measured in units of money $[C]=\mathbb{U}$. 
\end{itemize}
In other words, Kyle and Obizhaeva only use the subsequent hypothesis which is weaker than Assumption 1 above.
\begin{assumption}\label{ass:5unknowna} The market impact $G$ depends \textit{only} on the above 5 variables, i.e.
\begin{align}\label{eq:G_Ass}
G=g(Q,P,V,\sigma^2,C),
\end{align}
where the function $g: \R^5_+ \rightarrow \R_+$ as well as the quantity  $G$ are invariant under changes of the units chosen to measure the ``dimensions'' $\mathbb{T}, \mathbb{U}$ and $\mathbb{S}$.
\end{assumption}

Starting from this weaker assumption Kyle and Obizhaeva apply a similar reasoning as above, in particular the  argument of leverage neutrality. This leads to a system of four linear equations in five unknowns. The solution is not unique anymore, but leaves us with a degree of freedom which is expressed by the function $f$ below.
\begin{theorem}[Kyle and Obizhaeva] \label{thm:Kyle}
Under Assumptions~\ref{ass:mm} and \ref{ass:5unknowna}, the market impact is of the form
\begin{equation*}
	G = \frac{1}{L}f(Z),
\end{equation*}
where $f:\R_+\rightarrow\R_+$ is a function and the quantities $L$ and $Z$ are given by 
\begin{align}\label{eq:LZ}
L = \left(\frac{PV}{\sigma^2C} \right)^{1/3}\quad \text{and}\quad Z = \left(\frac{Q^3P^2\sigma^2}{VC^2} \right)^{1/3}.
\end{align}
\end{theorem}
A priori, the generality of the function $f:\R_+\to \R_+$ is not restricted by Assumption~\ref{ass:5unknowna}.
Specializing further as in the \emph{ansatz} \eqref{eq:ansatz}, one may assume $f$ to be of the form $f(z)=\operatorname{const}\cdot z^p$, for some $p\geq0$. This implies that $G=\text{const} \cdot Z^{p}/L$. In particular, the choice $p=1/2$ leads precisely to the relation \eqref{eq:SqrtLawIntro} obtained in Theorem~\ref{thm:SqrtLawIntro} above. Other choices of $p$ lead to different relations, some of them already considered in the literature. Moreover, we would like to emphasize that the quantities $L$ and $Z$ have a financially meaningful interpretation in terms of measuring liquidity and the size of meta-orders (see \cite{kyle2017dimensional}).
\newline

The roadmap of this note is as follows. In Section~\ref{sec:review}, we provide a brief review of the existing literature. Section~\ref{sec:dim_anal} introduces some notation as well as the so-called Pi-Theorem from dimensional analysis, which is the key to rigorously prove Theorems 1 and 2 in Section~\ref{sec:laws}. Section~\ref{sec:conc} concludes. Appendix~\ref{app:pendulum} discusses the example mentioned in the introduction, namely the period of a pendulum, in somewhat more detail, while some proofs are moved to Appendix~\ref{app:math}.

\section[Literature review]{Literature review}\label{sec:review}
As pointed out in recent reviews \cite{bouchaud2009markets, foucault2013market}, market impact can arise from different sources. 
For instance, Kyle in his seminal paper \cite{kyle1985continuous} derives from an agent-based model that market impact should be linear in the order size and permanent in time. The majority of studies, however, does not support this conclusion of Kyle's model. Instead, a body of  literature finds market impact being non-linear in the order size and fading in time, e.g. \cite{bouchaud2009markets}. In particular, the market impact is frequently found to be concave in the size of the meta-order and especially close to the square-root function, which causes the name \emph{square-root law} for market impact, see \cite{bacry2015market, bershova2013nonlinear, brokman2015market, engle2012measuring, gomes2015market, mastromatteo2014agent, moro2009market, toth2011anomalous}. Among other results, a market microstructure foundation in favor of the square-root law is provided in \cite{gabaix2006institutional}. The broad evidence for the square-root law relies on studies having data from different venues, maturities, historical periods and geographical areas and thus provides the square-root law with universality. On the other hand, it deserves to be mentioned that some studies reveal empirically deviations from the square-root law, e.g. \cite{almgren2005direct,zarinelli2015beyond}.

Let us try to elaborate on the relation between dimensional analysis and a general theory by alluding once more to the analogy with the period of the pendulum. Complementary to the introductory example, relation \eqref{eq:sqrt-pendulum} from physics can, of course, also be derived  from solving differential equations. Analogously, the square-root law for market impact can also be derived via solving partial differential equations, see \cite{donier2015consistent}. These authors formulate the dynamics of the average buy and sell volume density of the \emph{latent order book} in terms of partial differential equations under minimal model requirements. While the latent order book is a theoretical concept which records the trading intentions of market participants, traders typically do not display their true supply and demand, so that the fictitious, non-public latent order book differs from the observed limit order book. As we derive the square-root law for market impact via dimensional analysis, Theorem~\ref{thm:SqrtLawIntro} complements the existing literature.

\section[Some linear algebra]{Some linear algebra}\label{sec:dim_anal}
To review the basic results of dimensional analysis we follow Chapter 1 of the book by Bluman and Kumei \cite{bluman2013symmetries}. Additionally, the interested reader is referred to \cite{pobedrya2006proof} for a historical perspective and to \cite{curtis1982dimensional} for a purely mathematical treatment of dimensional analysis. We formalize the assumptions behind dimensional analysis in proper generality. However, for the purpose of the present paper we shall only need the degree of generality covered by Corollaries~\ref{thm:pi,k=0} and \ref{thm:pi,k=1} below.
\newline

\begin{assumption}[Dimensional analysis]\label{ass:DimAnal} ~ \begin{enumerate}[(i)]
\item Let the quantity of interest $U \in \R_+$ depend on $n$ quantities $W_1,\dots,W_n \in \R_+$, i.e.
\begin{align}
	U = H(W_1,W_2,\dots,W_n), \label{ProblemGeneral}
\end{align}
for some function $H:\R_+^n \to \R_+$.
\item The quantities $U,W_1,\dots,W_n$ are measured in terms of $m$ fundamental dimensions labelled $L_1,\dots,L_m$, where $m \leq n$. For any positive quantity $X$, its dimension $[X]$ satisfies $[X] = L^{x_1}_1\cdots L^{x_m}_m$ for some $x_1,\dots,x_m \in \R$. If $[X]=1$, the quantity $X$ is called \emph{dimensionless}.

The dimensions of the quantities $U,W_1,W_2,\dots,W_n$ are known and given in the form of vectors $a$ and $b^{(i)}\in\R^m$, $i = 1, \dots, n$, satisfying $[U] = L_1^{a_1}\cdots L_m^{a_m}$ and $[W_i] = L_1^{b_{1i}}\cdots L_m^{b_{mi}}$, $i = 1,\dots,n$.
Denote by $B = (b^{(1)},b^{(2)}, \dots,b^{(n)})$ the $m \times n$ matrix with column vectors $b^{(i)}= (b_{1i},\dots, b_{mi})^\top$, $i=1,\ldots,n$.
\item For the given set of fundamental dimensions  $L_1,\dots,L_m$, a \emph{system of units} is chosen in order to measure the value of a quantity.  A change from one system of units to another amounts to rescaling all considered quantities. In particular, dimensionless quantities remain unchanged and formula \eqref{ProblemGeneral} is invariant under arbitrary scaling of the fundamental dimensions.
\end{enumerate}
\end{assumption}
We can now state the main result from dimensional analysis (see \cite{bluman2013symmetries}).
\begin{theorem}[Pi-Theorem]
\label{thm:pi}
Under Assumption \ref{ass:DimAnal}, let $x^{(i)} := (x_{1i},\dots,x_{ni})^\top$, $i=1,\dots,k:= n-\operatorname{rank}(B)$ be a basis of the solutions to the homogeneous system $Bx=0$ and $y:= (y_1,\dots,y_n)^\top$ a solution to the inhomogeneous system $By = a$ respectively. Then, there is a function $F:\R_+^k \to \R_+$ such that 
\begin{align*}
	U  \cdot W_1^{-y_1} \cdots W_n^{-y_n} =  F(\pi_1,\dots,\pi_k),
\end{align*}
where $\pi_i := W_1^{x_{1i}} \cdots W_n^{x_{ni}}$ are dimensionless quantities, for $i=1,\dots,k$.
\end{theorem}

We shall only need the special cases $k=0$ and $k=1$, which are spelled out in the two subsequent corollaries.

\begin{corollary}
\label{thm:pi,k=0}
Under Assumption~\ref{ass:DimAnal}, suppose that $\operatorname{rank}(B)=n$ and let $y:=(y_1,\dots, y_n)^\top$ be the unique solution to the linear system $By = a$. Then there is a constant $\operatorname{const}>0$ such that 
\begin{align*}
U = \operatorname{const} \cdot\,W_1^{y_1} \cdots W_n^{y_n}.
\end{align*}
\end{corollary}

\begin{corollary}
\label{thm:pi,k=1}
Under Assumption~\ref{ass:DimAnal}, suppose that $\operatorname{rank}(B) = n - 1$ and let $x:= (x_1,\dots, x_n)^\top$ and $y:=(y_1,\dots, y_n)^\top$ be non-trivial solutions to the homogeneous and inhomogeneous systems $Bx = 0$ and $By = a$ respectively. Then there is a function $F: \R_+ \rightarrow \R_+$ such that
\begin{align*}
U = F(W_1^{x_{1}} \cdots W_n^{x_n}) W_1^{y_1} \cdots W_n^{y_n}.
\end{align*}
\end{corollary}

\section[Market impact]{Market impact}\label{sec:laws}

The aim of this section is to formalize and prove Theorems~\ref{thm:SqrtLawIntro} and \ref{thm:Kyle} stated in the introduction by applying Corollaries~\ref{thm:pi,k=0} and \ref{thm:pi,k=1}. In order to derive the market impact function from these corollaries, we need to formalize Assumption~\ref{ass:mm} in the framework of Section~\ref{sec:dim_anal}. Therefore, we take a closer look at this assumption and hence, the behavior of the quantities $G, Q, P, V, \sigma^2$ and $C$ in case of changing the firm's leverage defined below. From a conceptual point of view, the assumption of leverage neutrality gives an additional constraint on their behavior. This constraint can be understood as an additional though synthetic dimension in our analysis, which we refer to as the Modigliani-Miller ``dimension'' $\mathbb{M}$. The Modigliani-Miller dimension $\mathbb{M}$ of a share of a company can be measured in terms of the leverage $\mathcal{L}$, i.e. the quantity 
\begin{align*}
	\mathcal{L} = \frac{\text{total assets}}{\text{equity}}.
\end{align*}
Multiplying $\mathcal{L}$ by a factor $A > 1$  is equivalent to paying out $(1-A^{-1})$ of the equity as cash-dividends. On the other hand, multiplying $\mathcal{L}$ by a factor $0<A<1$ corresponds to raising new capital in order to increase the firm's own capital by $(A^{-1}-1)$ times its equity. Following \cite{kyle2017dimensional}, Assumption~\ref{ass:mm} can be reformulated in the following way:

\begin{assumption}[Leverage neutrality]
\label{ass:MMunit}
Scaling the Modigliani-Miller ``dimension'' $\mathbb{M}$ by a factor $A \in \R_+$ implies that 
\begin{itemize}
\item $Q$, $V$ and $C$ remain constant,
\item $P$ changes by a factor $A^{-1}$,
\item $\sigma^2$ changes by a factor $A^2$,
\item $G$ changes by a factor $A$.
\end{itemize}
\end{assumption}
To recapitulate in prose: Setting $A=2$ corresponds to paying out half of the equity as dividends in the sense that each share yields a dividend of $(1-A^{-1})P=P/2$. The stock price, thus, is multiplied by $A^{-1}=1/2$ while the volatility $\sigma$ and the percentage market impact $G$ are multiplied by $A=2$.
The remaining quantities are not affected by changing the leverage, in accordance with the insight of Modigliani and Miller \cite{modigliani1958cost} and the recent work by  Kyle and Obizhaeva \cite{kyle2017dimensional}. As indicated in the introduction, this assumption is referred to as \textit{leverage neutrality}.
\newline

We now reformulate Theorem~\ref{thm:SqrtLawIntro} by replacing the informally stated Assumption~\ref{ass:mm} by the more formal Assumption~\ref{ass:MMunit} and provide a proof.

\begin{theorem} \label{thm:SqrtLawLater}
	Under Assumptions~\ref{ass:4unknown} and \ref{ass:MMunit}, the market impact is of the form
	 	\begin{align} \label{eq:sqrtLaw}
			G &=\operatorname{const}\cdot\,\sigma \sqrt{\frac{Q}{V}},
		\end{align}
	 for some constant $\operatorname{const}>0$.
\end{theorem}

\begin{proof}
Combining Assumptions~\ref{ass:4unknown} and \ref{ass:MMunit} with the dimensions of $Q, P, V$ and $\sigma^2$ introduced in Section~\ref{sec:intro}, we obtain that the matrix $B$ and the vector $a$ are given by
\begin{align}\label{eq:dim1}
B = \left( \begin{array}{rrrr}
1 &-1 & 1 & 0   \\
0 & 1 & 0 & 0   \\
0 & 0 &-1 &-1   \\
0 &-1 & 0 & 2   \\
\end{array}
\right) \quad\text{and}\quad a = \left( \begin{array}{c}
0 \\ 0 \\ 0 \\ 1
\end{array} \right).
\end{align}
Table~\ref{tab:DimensionTabel} summarizes how $B$ and $a$ can be derived and should be read as follows: Assume the measurement of dimension $\mathbb{S}$ referring to the unit of shares is rescaled by a factor $S$, then $Q$ changes by a factor $S$, $P$ changes by a factor $S^{-1}$, $V$ changes by a factor $S$, while $\sigma^2$ does not change. Likewise, the last row labelled by $\mathbb{M}$ indicates that if the 
leverage $\mathcal{L}$ of the firm
changes by a factor $A$, $Q$ does not change, $P$ changes by a factor $A^{-1}$, and so on.

As the matrix $B$ has full rank, i.e. rank$(B) = 4 = n$, and Assumption~\ref{ass:DimAnal} is satisfied, applying Corollary~\ref{thm:pi,k=0} yields
\begin{align*}
	G = \text{const}\cdot Q^{y_1}P^{y_2}V^{y_3}\sigma^{2y_4},
\end{align*}
for some constant const$ >0$, where $y=(y_1, y_2, y_3, y_4)^\top$ is the unique solution of the linear system $By=a$ which is given by $y=(\frac{1}{2}, 0, -\frac{1}{2}, \frac{1}{2})^\top$.
\end{proof}

Theorem~\ref{thm:SqrtLawLater} implies the well known \emph{square-root law} for market impact. We would like to highlight that the present derivation of the square-root law does not rely on economic, empirical or theoretical assumptions except the dependence of $G$ on $Q, P, V$ and $\sigma^2$ only, as well as leverage neutrality. Donier et al. \cite{donier2015consistent} present an alternative derivation of the square-root law relying on partial differential equations.
\newline

As discussed above, Kyle and Obizhaeva \cite{kyle2017dimensional} consider yet another variable to influence the market impact, namely the ``bet cost'' $C$ leading to the weaker Assumption~\ref{ass:5unknowna}. An economic motivation to include $C$ in the analysis is provided also in \cite{kyle2016market}. Based on Assumption~\ref{ass:5unknowna}, Kyle and Obizhaeva \cite{kyle2017dimensional} derive a more general result summarized in Theorem~\ref{thm:Kyle}. The methodology of Section~\ref{sec:dim_anal} can be employed to prove a similar result stated below. 

\begin{table}
\begin{center}
\begin{tabular}{c|rrrr:r|r}
 & $Q$ & $P$ & $V$ & $\sigma^2$ & $C$ &$G$ \\ \hline
$\mathbb{S}$ & 1 &-1 & 1 & 0 & 0 & 0\\
$\mathbb{U}$ & 0 & 1 & 0 & 0 & 1 & 0\\
$\mathbb{T}$ & 0 & 0 &-1 &-1 & 0 & 0 \\
$\mathbb{M}$ & 0 &-1 & 0 & 2 & 0 & 1\\
\end{tabular}
\caption{A labelled overview of the matrix $B$ related to the dimensions of the quantities $(Q,P,V,\sigma^2)$ and the matrix $K$ related to the dimensions of the quantities $(Q,P,V,\sigma^2,C)$ respectively, as well as the vector $a$ related to the dimensions of $G$.}
\label{tab:DimensionTabel}
\end{center}
\end{table}

First of all, the matrix $B$ used in the proof of Theorem~\ref{thm:SqrtLawLater} is extended by one column, which corresponds to $C$, to obtain the matrix 	
\begin{align}\label{eq:Kanda}
K = \left( \begin{array}{rrrrr}
1 &-1 & 1 & 0 & 0 \\
0 & 1 & 0 & 0 & 1 \\
0 & 0 &-1 &-1 & 0 \\
0 &-1 & 0 & 2 & 0 
\end{array}
\right).
\end{align}
The vector $a = \left(0, 0, 0, 1\right)^{\top}$ defined in \eqref{eq:dim1} remains unchanged. Table~\ref{tab:DimensionTabel} illustrates how the additional variable $C$ is related to the considered ``dimensions''.

Let us compute the solution space ${\cal H}$ of the homogeneous system $Kx=0$, which is given by the kernel of the linear map induced by the matrix $K$, as well as the solution space ${\cal I}$ of the inhomogeneous linear system $Ky=a$:
\begin{equation*}
\mathcal{H} = \left\lbrace \lambda \left( \begin{array}{r}
3 \\ 2 \\ -1 \\ 1 \\ -2
\end{array} \right), \lambda \in \R \right\rbrace
\quad
\text{and}
\quad	
\mathcal{I} = \left\lbrace  \left( \begin{array}{r}
-1 \\ -1 \\ 0 \\ 0 \\ 1
\end{array} \right)+ \lambda\left( \begin{array}{rrrrr}
3 \\ 2 \\ -1 \\ 1 \\ -2
\end{array} \right), \lambda \in \R \right\rbrace.
\end{equation*}
\begin{theorem}
\label{thm:generalLater}
Suppose Assumptions~\ref{ass:5unknowna} and \ref{ass:MMunit} hold. Fix $x=(x_1,\ldots,x_5)^{\top} \in \mathcal{H}$ and $y=(y_1,\ldots,y_5)^{\top} \in \mathcal{I}$. There is a function $f:\R_+\to \R_+$ such that
\begin{align} \label{eq:GeneralFunG}
	G = Q^{y_1} P^{y_2} V^{y_3} \sigma^{2y_4} C^{y_5} f\left(Q^{x_1} P^{x_2} V^{x_3} \sigma^{2x_4} C^{x_5}\right).
\end{align}
\end{theorem}
\begin{proof}
Combining Assumptions~\ref{ass:5unknowna} and \ref{ass:MMunit} with the dimensions of $Q, P, V$, $\sigma^2$ and $C$ introduced in Section~\ref{sec:intro}, we recover the matrix $K$ and the vector $a$ given in \eqref{eq:Kanda} and \eqref{eq:dim1} respectively. Since Assumption~\ref{ass:DimAnal} is satisfied and $\operatorname{rank}(K)=4$, applying Corollary~\ref{thm:pi,k=1} completes the proof.
\end{proof}

For example, by setting
\begin{align*}
	x = \left(1, \frac{2}{3}, -\frac{1}{3}, \frac{1}{3}, -\frac{2}{3}\right)^\top  \in \mathcal{H} \qquad \text{and} \qquad y = \left(0, -\frac{1}{3}, -\frac{1}{3}, \frac{1}{3},\frac{1}{3}\right)^\top \in \mathcal{I},
\end{align*}
Theorem~\ref{thm:generalLater} yields precisely the formula of Theorem~\ref{thm:Kyle} given in the introduction, i.e. 
\begin{align} \label{eq:KyleFunG}
	G &= \left( \frac{\sigma^2 C}{P V}\right)^{1/3} f  \left( \left(\frac{Q^3P^2\sigma^2}{VC^2} \right)^{1/3} \right)  \notag \\
&= \frac{1}{L} f\left(Z\right),
\end{align}
where $L$ and $Z$ are defined in \eqref{eq:LZ}. One may also consider other choices for $x \in \mathcal{H}$ and $y \in \mathcal{I}$, for example:
\begin{equation*}
	x = \left(3, 2, -1, 1, -2\right)^\top  \in \mathcal{H} \qquad \text{and} \qquad y = \left(\frac{1}{2}, 0,-\frac{1}{2}, \frac{1}{2}, 0\right)^\top \in \mathcal{I}.
\end{equation*}
Formula \eqref{eq:GeneralFunG} then takes the form
\begin{align}
	G = \sigma \sqrt{\frac{Q}{V}} \ h \left(Z^3 \right) = \sigma \sqrt{\frac{Q}{V}} \ h \left(\frac{Q^3 P^2 \sigma^2}{V C^2} \right) .
	\label{eq:sqrt_gene}
\end{align}
If the function $h$ in \eqref{eq:sqrt_gene} is not a constant, this formula describes nicely the deviation from the square-root law \eqref{eq:sqrtLaw} in a multiplicative way.

\begin{remark}
It is important to note that \eqref{eq:KyleFunG} as well as \eqref{eq:sqrt_gene} are both \emph{the general solution} of the functional relation described by Theorem \ref{thm:generalLater} and therefore coincide.
The difference is that the (arbitrary) functions $f$ and $h$ are not identical, but rather in a one-to-one relation when passing from \eqref{eq:KyleFunG} to \eqref{eq:sqrt_gene}.
\end{remark}

As pointed out by Kyle and Obizhaeva \cite{kyle2017dimensional}, different choices of $f$ in equation~\eqref{eq:KyleFunG} (resp. $h$ in \eqref{eq:sqrt_gene}) lead to some particularly relevant market impact models studied in the literature.
\begin{enumerate}[(a)]
\item The proportional market impact: $f\equiv \operatorname{const}$ (resp. $h(x)=\operatorname{const} \cdot \ x^{-1/6}$)  leads to
\begin{align*}
G = \text{const} \cdot\left(\frac{\sigma^2 C}{P V}\right)^{1/3}.
\end{align*}
\item The square-root impact: $f(z)=\operatorname{const}\cdot z^{1/2}$ (resp. $h\equiv \operatorname{const}$) leads to
\begin{align*}
G &=\text{const}\cdot \sigma \sqrt{\frac{Q}{V}},
\end{align*}
the unique solution which does not depend on $C$.
\item The linear market impact: $f(z)=\operatorname{const} \cdot z$ (resp. $h(x)=\operatorname{const} \cdot \ x^{1/6}$) leads to
\begin{align*}
G = \text{const} \cdot Q \left(\frac{\sigma^4 P}{C V^2}\right)^{1/3}.
\end{align*}

\end{enumerate}
\begin{remark}
	It is also important to note that only two properties of the variable $C$ enter the above dimensional analysis: the ``dimension'' of $C$ is money, i.e. $[C]=\mathbb{U}$, and $C$ remains unchanged by scaling the Modigliani-Miller ``dimension'' $\mathbb{M}$ by a factor $A\in\R_+$. The above result, therefore, does not rely on the interpretation of the quantity $C$ as ``bet cost'' as considered in Kyle and Obizhaeva \cite{kyle2017dimensional}, but applies to any other quantity with the two aforementioned properties just as well.
\end{remark}

For example, an interesting alternative to $C$ enjoying these properties can be found in the work on the intraday trading invariance hypothesis by Benzaquen et al. \cite{benzaquen2016unravelling}. Rather than $C$, these authors consider the spread cost $\mathcal{C}$, which can be interpreted as the transaction cost incurred by trading $Q$ shares. More formally, denote by $S$ the bid-ask spread measured in units of money per share $[S]=\mathbb{U}/\mathbb{S}$.\footnote{
%
It should be noticed that the spread $S$ remains unchanged when scaling the Modigliani-Miller dimension $\mathbb{M}$ by a factor $A\in\mathbb{R}_+$. For instance, this can be inferred from an argument of Kyle and Obizhaeva \cite[Section 3: Empirical Evidence based on Bid-Ask Spreads and $\ldots$]{kyle2017dimensional}, whose empirical analysis uses that $G$ and $S/P$ have the same dimensional properties. Since the concept of leverage neutrality tells us precisely how $G$ and $P$ change when $\mathbb{M}$ is scaled by a factor $A\in\mathbb{R}_+$, the spread $S$ has to remain unchanged.} The spread cost $\mathcal{C}$ of a meta-order with size $Q$ is then defined by $\mathcal{C}:=Q S$ and hence measured in units of money $[\mathcal{C}]=\mathbb{U}$. Thus, the mathematical analysis above remains \emph{totally unchanged} when $C$ is replaced  by $\mathcal{C}$. 
In particular, the analogue of formula \eqref{eq:sqrt_gene} then reads as 
\begin{align}\label{eq:sqrt_spread}
	G = \sqrt{\frac{Q}{V}} \ h \left( \frac{Q^3P^2\sigma^2}{V \mathcal{C}}^2 \right).
\end{align}

To finish this section, we shall consider one more possible set of 5 explanatory variables which will lead us into a somewhat different direction. Instead of $C$ (or any appropriate alternative such as $\mathcal{C}$), we consider a variable with a different dimension, namely the length of the time interval $[0, T]$ over which the meta-order is executed. Clearly, the length $T$ is an obvious candidate to influence the market impact:
\begin{itemize}
	\item $T\quad\,$ the length of the execution interval, measured in units of time $[T]=\mathbb{T}$. 
\end{itemize}
In practice, the interval length $T$ can vary from a fraction of hours up to several days or even weeks. In an analogous manner as the ``bet cost'' $C$ enters Assumption~\ref{ass:5unknowna}, the subsequent assumption incorporates the length of the execution interval.

\begin{assumption}\label{ass:5unknown_time} The market impact $G$ depends \textit{only} on the variables $Q$, $P$, $V$, $\sigma^2$ and $T$, i.e.
\begin{align}\label{eq:G_Ass_time}
G=g(Q,P,V,\sigma^2,T),
\end{align}
where the function $g: \R^5_+ \rightarrow \R_+$ as well as the quantity $G$ are invariant under changes of the units chosen to measure the ``dimensions'' $\mathbb{T}, \mathbb{U}$ and $\mathbb{S}$.
\end{assumption}
Playing a similar game as above, the matrix $B$ used in the proof of Theorem \ref{thm:SqrtLawLater} is extended by one column corresponding to $T$ and now given by
\begin{align*}
K = \left( \begin{array}{rrrrr}
1 &-1 & 1 & 0 & 0 \\
0 & 1 & 0 & 0 & 0 \\
0 & 0 &-1 &-1 & 1 \\
0 &-1 & 0 & 2 & 0 
\end{array} \right).
\end{align*}
The vector $a$ given in \eqref{eq:dim1} remains unchanged.
The solution spaces ${\cal H}$ and ${\cal I}$ of the homogeneous system $Kx=0$ and the inhomogeneous system $Ky=a$ respectively, are given by
\begin{equation*}
	\mathcal{H} = \left\lbrace \lambda \left( \begin{array}{r}
-1 \\ 0 \\ 1 \\ 0 \\ 1
\end{array} \right), \lambda \in \R \right\rbrace
\quad
\text{and}
\quad 
\mathcal{I}= 
 \left\lbrace  \left( \begin{array}{r}
\frac{1}{2} \\ 0 \\ -\frac{1}{2} \\ \frac{1}{2} \\ 0
\end{array} \right)+ \lambda\left( \begin{array}{r}
1 \\ 0 \\ -1 \\ 0 \\ -1
\end{array} \right), \lambda \in \R \right\rbrace.
\end{equation*}
Under the assumptions of leverage neutrality and the exclusive dependence of the market impact $G$ on $Q$, $P$, $V$, $\sigma^2$ and $T$, dimensional analysis leads to the following result.

\begin{theorem}
\label{thm:time}
Suppose Assumptions~\ref{ass:MMunit} and \ref{ass:5unknown_time} hold. Fix $x=(x_1,\ldots,x_5)^{\top} \in \mathcal{H}$ and $y=(y_1,\ldots,y_5)^{\top} \in \mathcal{I}$. There is a function $f:\R_+\to \R_+$ such that
\begin{align*}
	G = Q^{y_1} P^{y_2} V^{y_3} \sigma^{2y_4} T^{y_5} f\left(Q^{x_1} P^{x_2} V^{x_3} \sigma^{2x_4} T^{x_5}\right).
\end{align*}
\end{theorem}

\begin{proof}
Since Assumption~\ref{ass:DimAnal} is satisfied and $\operatorname{rank}(K)=4$, the result follows by Corollary~\ref{thm:pi,k=1}.
\end{proof}
For instance, setting
\begin{equation*}
	x = \left(1, 0,-1,0,-1\right)^\top  \in \mathcal{H} \qquad \text{and} \qquad y = \left(\frac{1}{2}, 0,-\frac{1}{2}, \frac{1}{2}, 0\right)^\top \in \mathcal{I},
\end{equation*}
we obtain
\begin{equation}
\label{eq:GeneralSqrtLaw_time}
	G= \sigma\sqrt{\frac{Q}{V}} \ h\left(\frac{Q}{VT} \right).
\end{equation}
Similar to \eqref{eq:sqrt_gene}, where the ``bet cost'' $C$ appears as a quantity influencing the deviation of \eqref{eq:sqrt_gene} from  \eqref{eq:sqrtLaw}, the function $h$ in \eqref{eq:GeneralSqrtLaw_time} characterizes the deviation from the square-root law \eqref{eq:sqrtLaw} in dependence of the length of the execution interval $T$.

Donier et al. \cite{donier2015consistent} derive the square-root law based on a model taking the execution horizon $T$ into account. Thus, the question arises under which conditions on $T$ we recover the square-root law \eqref{eq:sqrtLaw}. The answer is simple: If the length of the execution interval $T$ depends exclusively on either or all of the quantities $Q$, $V$, $P$ and $\sigma^2$, Assumption~\ref{ass:5unknown_time} can be replaced by Assumption~\ref{ass:4unknown} and we are back in the setting of Theorem \ref{thm:SqrtLawIntro}, where the market impact obeys the square-root law. In practice, the condition that $T$ depends on $Q$, $P$, $V$ and $\sigma^2$ can be satisfied in case the investor determines the execution horizon $T$ according to the latter quantities.

\section[Conclusion]{Conclusion} \label{sec:conc}
The main contribution of this paper is a derivation of the \emph{square-root law} for market impact. The strong empirical support in favor of this \emph{law} provides it with a universal character. Inspiring for our work, Kyle and Obizhaeva \cite{kyle2017dimensional} derive a general form for the market impact function relying on dimensional analysis as well as the concepts of leverage neutrality and market microstructure invariance, where the \emph{square-root law} turns out to be a special case of their result. Complementary to their approach, we present a direct and simple derivation of the \emph{square-root law} by requiring only two assumptions: Firstly, the market impact of a given meta-order \emph{only} depends on its size, the corresponding stock price, the traded volume in the stock as well as its volatility. Secondly, we employ the concept of leverage neutrality as in \cite{kyle2017dimensional}. This idea is in line with the Modigliani-Miller invariance principle \cite{modigliani1958cost} and explains how the considered quantities behave when changing the leverage of a firm. Relying on these plausible assumptions, we apply dimensional analysis in a rigorous way to show that the market impact of a meta-order is proportional to the volatility as well as to the square-root of this order's size and inversely proportional to the square-root of the traded volume.

We also discuss several extensions of this result by including the following quantities as additional explanatory variables: the ``bet cost'' $C$ \eqref{eq:sqrt_gene} as in \cite{kyle2017dimensional}, the ``spread cost'' $\mathcal{C}$ \eqref{eq:sqrt_spread} as in \cite{benzaquen2016unravelling}, or the length $T$ of the execution interval \eqref{eq:GeneralSqrtLaw_time}.

\appendix

\section[The pendulum]{The pendulum}\label{app:pendulum}

In the setting of Theorem~\ref{thm:SqrtLawIntro}, somehow surprisingly, the market impact does not depend on the stock price, although -- a priori -- the price is included in our analysis. There is a simple explanation: The stock price is the only quantity in our analysis involving the ``dimension'' $\mathbb{U}$ of money.
Hence, in the setting of Theorem \ref{thm:SqrtLawIntro} it cannot play a role, because the market impact also does not involve the ``dimension'' $\mathbb{U}$ of money. 
In the following, we give a more detailed discussion of this argument in the form of an analogy to the case of the pendulum. In this example, the period also does not depend on the mass of the pendulum which - a priori - is considered as an explaining variable.

Consider a pendulum with length $l$ (measured in meters), mass $m$ (measured in grams) and period (measured in seconds).
Assume that the period depends only on $l, m$ and the acceleration $g$ of gravity (measured in meters per seconds squared).
That is, we assume that there is a function $f:\R_+^3\to \R_+$ such that 
\begin{equation*}
	\text{period} = f(l,m,g).
\end{equation*}
From Table~\ref{tab:DimensionPendulum}, we get the matrices
\begin{align*}
D = \left( \begin{array}{rrr}
1 &0 & 1   \\
0 & 1 & 0   \\
0 & 0 &-2   \\
\end{array}
\right) \quad\text{and}\quad c = \left( \begin{array}{c}
0 \\  0 \\ 1
\end{array} \right),
\end{align*}
which represent the dimensions of the quantities $(l, m, g)$ and the period respectively in the unit system meter, gram and seconds. As $D$ has full rank, it follows from Corollary~\ref{thm:pi,k=0} that 
\begin{equation*}
\text{period} = \operatorname{const}\cdot \ l^{y_1}m^{y_2}g^{y_3},
\end{equation*}
for some $\text{const} >0$, where the unique solution of the linear system $Dy = c$ is given by $y=(\frac{1}{2}, 0, -\frac{1}{2})^{\top}$.
Thus,
\begin{equation}
\label{eq:pendulum-first-sol}
\text{period} = \operatorname{const}\cdot \ \sqrt{\frac{l}{g}}.
\end{equation}
Why does this solution not involve the variable $m$? The answer is given by looking at the second row of $D$ and the second coordinate of $c$, which forces $y_2$ to equal zero. This is perfectly analogous to the role of the variable $P$, i.e. the price of the stock, in the setting of Theorem \ref{thm:SqrtLawIntro}.
\newline

Next, we shall illustrate the difference between Theorem \ref{thm:SqrtLawIntro} and \ref{thm:Kyle} by discussing an analogous variation of the assumptions in the case of the pendulum.
The crucial assumption in the reasoning above was that the period of the pendulum is completely determined by its length, its mass and the acceleration due to gravity. However, it is conceivable (and, in fact, the case) that the period also depends on other variables, e.g. the amplitude $a$ (measured in meters) of the observed swing. In other words, we might also start from the weaker assumption 
\begin{align}\label{eq:weaker}
\text{period} = f(l,m,g,a).
\end{align}
The matrix containing the dimensions of the observed quantities is now given by 
\begin{align*}
\widetilde{D} = \left( \begin{array}{rrrr}
1 &0 & 1&1   \\
0 & 1 & 0 &0   \\
0 & 0 &-2 &0  \\
\end{array}
\right).
\end{align*}
It follows from Corollary~\ref{thm:pi,k=1} that the general form of the relation \eqref{eq:weaker} is given by
\begin{equation*}
\text{period} = l^{y_1}m^{y_2}g^{y_3}a^{y_4}h(l^{x_1}m^{x_2}g^{x_3}a^{x_4}),
\end{equation*}
for some function $h:\R_+\to \R_+$, where $x = (x_1, x_2, x_3,x_4)^{\top}$ is a solution of the homogeneous system $\widetilde{D}x = 0$ and $y= (y_1, y_2, y_3, y_4)^{\top}$ a solution of the inhomogeneous system $\widetilde{D}y = c$. Choosing
\begin{equation*}
	x = \left(\frac{1}{2}, 0,-\frac{1}{2},0\right)^\top  \qquad \text{and} \qquad y = \left(1, 0,0, -1 \right)^\top,
\end{equation*}
we obtain
\begin{align}
\text{period} = \sqrt{\frac{l}{g}} \ h\left(\frac{l}{a}\right). \label{eq:pendulum-general}
\end{align}
In the setting of \eqref{eq:weaker}, dimensional analysis does not allow to determine the function $h$.
In order to do so, we need some additional information. In physics we have the possibility of experiments. Already Galileo noticed the -- at first glance surprising -- experimental result that the period of the pendulum does \emph{not} (at least not strongly) depend on the amplitude. Using this insight from experimental physics, we conclude that $h \equiv \operatorname{const}$ is a  physically reasonable choice in the general solution \eqref{eq:pendulum-general}.
\begin{table}
\begin{center}
\begin{tabular}{c|rrr:r|c}
	 & $l$ & $m$ & $g$ & $a$ & period\\
\hline
lenght & 1 & 0 & 1 & 1 & 0 \\
mass   & 0 & 1 & 0 & 0 & 0 \\
time   & 0 & 0 &-2 & 0 & 1  \\
\end{tabular}
\caption{A labelled overview of the matrix summarizing the dimensions of the quantities considered to determine the period of a pendulum.}
\label{tab:DimensionPendulum}
\end{center}
\end{table}
\newline

We conclude this discussion by making the analogy to the case of Theorems \ref{thm:SqrtLawIntro} and \ref{thm:Kyle} above.
Dimensional analysis alone cannot decide whether the special solution given by Theorem \ref{thm:SqrtLawIntro} is the ``true'' relation between the market impact and the relevant variables, or whether some other explanatory variables as provided, e.g., by \eqref{eq:sqrt_gene}, \eqref{eq:sqrt_spread}, or \eqref{eq:GeneralSqrtLaw_time} yield the ``true'' relation.
To answer this question one has to take recourse either to economic theory or to empirical analysis.
This is analogous to the above discussed situation of the pendulum where \emph{physical experiments} yield that the special case \eqref{eq:pendulum-first-sol} of the more general solution \eqref{eq:pendulum-general} is in fact the ``true" relation (as long as the amplitude remains within reasonable bounds).

To round up this discussion, we give an example how dimensional analysis can lead astray, if applied blindly. Start with the (silly) assumption that the period of the pendulum depends \emph{only} on the mass $m$, the acceleration $g$, and the amplitude $a$ so that 
\begin{equation*}
	\text{period} = f(m,g,a),
\end{equation*}
for some function $f:\R^3_+\rightarrow\R_+$. Repeating verbatim the analysis preceding \eqref{eq:pendulum-first-sol} we obtain 
\begin{equation}\label{eq:pendulum-silly}
	\text{period} = \operatorname{const} \cdot \sqrt{\frac{a}{g}},
\end{equation}
as the unique solution satisfying the invariance properties of dimensional analysis. But, of course, the solution \eqref{eq:pendulum-silly} is far from physical reality. The reason is that we have chosen \emph{a wrong set of explanatory variables}. In other words, dimensional analysis only yields reasonable solutions if the set of explanatory variables is well chosen and really contains essentially all the information necessary to determine the quantity of interest.

\section[Appendix]{Proof of the Pi-Theorem}\label{app:math}

\textbf{Proof of Theorem~\ref{thm:pi}:}
We pass to logarithmic coordinates by using the following notation: Given $Z \in \R_+$ we shall write $\widetilde{Z} = \log(Z)$. On the logarithm scale $\widetilde{U}$ 
	satisfies
	\begin{equation}
	\label{eq:utilde}
		\widetilde{U} = g\left(\widetilde{W}_1,\dots,\widetilde{W}_n\right),
	\end{equation}
	for some function $g:\R^n\to \R$.
	
	If the left and the right hand sides of \eqref{eq:utilde} do not depend on $L_1$, then it is sufficient to work with the units $(L_2,\dots,L_m)$.
	On the other hand, if $a_1\neq 0$ and $b_{11}=\dots=b_{1n}=0$ then $g\equiv 0$.
	If there exists $i \in \lbrace 1, \dots,n \rbrace$ such that $b_{1i}\neq 0$, we assume without loss of generality that $b_{11}\neq 0$.
	Putting $\widetilde{V}:= -\frac{a_1}{b_{11}}\widetilde{W}_1 + \widetilde{U}$ and $\widetilde{X}_{i-1}:= -\frac{b_{1i}}{b_{11}}\widetilde{W}_1+ \widetilde{W}_i$, $i=2,\dots,n$ we have
	\begin{align*}
		\widetilde{V} &= -\frac{a_1}{b_{11}}\widetilde{W}_1 + g\left(\widetilde{W}_1, \frac{b_{1i}}{b_{11}}\widetilde{W}_1+ \widetilde{X}_1,\dots,\frac{b_{1n}}{b_{11}}\widetilde{W}_1+ \widetilde{X}_{n-1} \right)\\
		          &= f\left( \widetilde{W}_1,\widetilde{X}_1,\dots,\widetilde{X}_{n-1}\right),
	\end{align*}
	for some function $f$.
	Let $\lambda \in \R$ and put $L^*_1:= e^{\lambda} L_1$.
	Since the dimensions of $\widetilde{V}$ and $ \widetilde{X}_{i-1}$, $i=2,\dots,n$ are given in terms of the units  $(L_1,\dots, L_m)$, the quantities $\widetilde{V}, \widetilde{X}_{1}, \dots, \widetilde{X}_{n-1}$ remain unchanged  upon passing to the system of units $(L^*_1,L_2,\dots,L_m)$. 
	On the other hand, $\log\left([{W}_1]\right) = -\lambda b_{11} + b_{11}\widetilde{L}^*_1 + \sum_{i=2}^mb_{i1}\widetilde{L}_i $  so that in the system of units $(L^*_1,L_2,\dots ,L_m)$ it holds
	\begin{equation*}
		\widetilde{V} = f\left( \lambda b_{11}+\widetilde{W}_1,\widetilde{X}_1,\dots,\widetilde{X}_{n-1}\right).
	\end{equation*}
	Since $\lambda$ was taken arbitrary, $f$ does not depend on the first component, that is, 
	\begin{equation}
	\label{eq:V-one-time}
		\widetilde{V} = f\left(\widetilde{X}_1,\dots,\widetilde{X}_{n-1}\right).
	\end{equation}
	By repeating the argument rank$(B)-1$ times, we obtain
	\begin{equation}
	\label{eq:utilde.lhs}
		\widetilde{U} - \sum_{j=1}^ny_j\widetilde{W}_j = h\left(\sum_{j=1}^nx_{1j}\widetilde{W}_j,\dots,\sum_{j=1}^nx_{kj}\widetilde{W}_j\right).
	\end{equation}
	In fact, since $x^{(i)}$ is a solution of the homogeneous system $Bx = 0$, the quantity $\sum_{j=1}^nx_{1j}\widetilde{W}_j$ is dimensionless.
	Notice that $Bx=0$ has $k = n - \text{rank}(B)$ linearly independent solutions.
	Similarly, since $y$ is a solution of the inhomogeneous system $By=0$, the left hand side in \eqref{eq:utilde.lhs} is a dimensionless quantity.
	Hence, there is a function $F:\R_+^k\to \R_+$ such that
	\begin{equation*}
		U\cdot W_1^{-y_1}\cdots W_n^{-y_n} = F(\pi_1,\dots,\pi_n),
	\end{equation*}
	with $\pi_j:= W_1^{x_{1j}}\cdots W_n^{x_{nj}}$, $j=1,\dots,k$.
%
\hfill $\square$\vspace*{3mm} \\
\textbf{Proof of Corollary~\ref{thm:pi,k=0}:}
First notice that if $k=0$ then $n = \text{rank}(B)$.
Repeating the argument leading to \eqref{eq:V-one-time} rank$(B)-2$ times, we can find $(z_1,\dots,z_{n-1})^{\top}\in  \R^{n-1}$, a quantity $X$ with dimension $[X]= L^{\alpha_m}_m$ and a function $f:\R\to \R$ such that
\begin{equation*}
\widetilde{Y}:=	\widetilde{U} -\sum_{j=1}^{n-1}z_j\widetilde{W}_j = f(\widetilde{X}).
\end{equation*}
Let us denote by $L^{c_m}_m$ the dimension of $Y$.
We can assume without loss of generality that $\alpha_m\neq 0$.
As in the proof of Theorem \ref{thm:pi}, we have
\begin{equation*}
	\widetilde{V}:=\widetilde{U} - \sum_{j=1}^{n-1}z_j\widetilde{W}_j - \frac{c_m}{\alpha_m}X = - \frac{c_m}{\alpha_m}X+ f(\widetilde{X}) = g(\widetilde{X}),
\end{equation*}
for some function $g$.
Let $\lambda \in \R$ and put ${L}^*_m = e^\lambda {L}_m $.
Since $\widetilde{V}$ is dimensionless, its value does not change when passing to the unit $L^*_m$.
On the other hand, $\log([X]) = -\alpha_m\lambda + \alpha_mL^*_m $.
Hence, with respect to the unit $L^*_m$ we have
\begin{equation*}
\widetilde{V} = g(\alpha_m\lambda + X).
\end{equation*}
Since $\lambda$ was taken arbitrary, the function $g$ must be a constant.
Thus, there is $\text{const}>0$ such that $U=\text{const}\cdot W_1^{y_1} \cdots W_n^{y_n}$, since the right hand side of the latter equation has the dimension of $U$, if and only if $By = a$.
\hfill $\square$\vspace*{3mm} \\
\textbf{Proof of Corollary~\ref{thm:pi,k=1}:}
The result follows from a direct application of Theorem~\ref{thm:pi}.
\hfill $\square$


\end{document}